\pdfoutput=1
\newif\ifdraft \draftfalse 

\drafttrue

\documentclass{article}
\usepackage{graphicx} 
\usepackage[numbers]{natbib}
\usepackage{fullpage}

\usepackage{macros}
\usepackage{xspace}
\usepackage{xcolor}
\definecolor{darkblue}{rgb}{0,0,.5}
\usepackage[colorlinks=true,allcolors=darkblue]{hyperref}
\usepackage{bbm}
\usepackage{tikz}
\usetikzlibrary{matrix}
\usepackage{overpic}
\usepackage{makecell}
\usetikzlibrary{patterns}
\usepackage{authblk}
\usepackage{enumerate}
\usepackage{amssymb}
\usepackage{amsbsy}
\usepackage{amsmath}
\usepackage{amsthm}
\usepackage{amsfonts}
\usepackage{algorithm}
\usepackage{algpseudocode}

\usepackage{latexsym}
\usepackage{graphicx}
\usepackage{color}
\usepackage{xifthen}

\usepackage{color-edits}
\addauthor{Ryan}{red}

\usepackage{latexsym}

\usepackage{subcaption}

\newcommand{\R}{\mathbb{R}}

\newcommand{\N}{\mathbb{N}}

\newtheorem{theorem}{Theorem}
\newtheorem{lemma}{Lemma}[section]

\newtheorem{definition}{Definition}[section]

\title{A Unifying Privacy Analysis Framework for Unknown Domain Algorithms in Differential Privacy}
\author{Ryan Rogers}

\begin{document}

\maketitle

\begin{abstract}
There are many existing differentially private algorithms for releasing histograms, i.e. counts with corresponding labels, in various settings.  Our focus in this survey is to revisit some of the existing differentially private algorithms for releasing histograms over unknown domains, i.e. the labels of the counts that are to be released are not known beforehand.  The main practical advantage of releasing histograms over an unknown domain is that the algorithm does not need to fill in missing labels because they are not present in the original histogram but in a hypothetical neighboring dataset could appear in the histogram.  However, the challenge in designing differentially private algorithms for releasing histograms over an unknown domain is that some outcomes can clearly show which input was used, clearly violating privacy.  The goal then is to show that the differentiating outcomes occur with very low probability.  We present a unified framework for the privacy analyses of several existing algorithms.  Furthermore, our analysis uses approximate concentrated differential privacy from \citet{BunSt16}, which can improve the privacy loss parameters rather than using differential privacy directly, especially when composing many of these algorithms together in an overall system.  
\end{abstract}

\section{Introduction}

Releasing histograms, counts over some set of items, is one of the most fundamental tasks in data analytics.  Given a dataset, we want to compute the number of rows that satisfy some condition, grouped together by some column(s) of interest, say the number of rows in each country.  Despite the commonality of this task, there are many different differentially private algorithms to use, depending on the setting we are in, e.g. do we want to show only the top-$k$, do we know how many items a user can have in any row, do we know all possible values that a column can take, how often is the data refreshed?  We will focus on the setting where we do not have, or even know, the counts of all possible items in a histogram.  This is typical in practice, because SQL queries will only return items that have positive counts, otherwise how would it know items not present in the dataset?  Unfortunately, differential privacy requires considering a hypothetical neighboring dataset, which might contain items that were previously unseen, which makes the privacy analysis challenging.  Either we need to populate all possible items that \emph{could} appear in the dataset and fill in the missing counts as zero, or we need to design better, more practical, DP algorithms.  In the latter case, we would want to be able to ensure DP whether we were computing the top-10 skills in a dataset or the top-10 credit card numbers in the dataset.  The latter should return no results, but the algorithm should ensure privacy in both cases.  We will refer to the scenario where the DP algorithm does not know the domain of items that could be in a dataset beforehand as the \emph{Unknown Domain} setting.

In differential privacy there are typically two different parameters $\diffp, \delta$ that are used to quantify the privacy of an algorithm.  The parameter $\diffp$ is typically referred to as the amount of \emph{privacy loss} that an algorithm ensures, while $\delta$ is commonly referred to as the probability in which the algorithm can have larger privacy loss than $\epsilon$.  In fact, when $\delta >0$, we say an algorithm satisfies \emph{approximate} differentially privacy, while we say \emph{pure} DP with $\delta = 0$. 
 The parameter $\delta$ then can be the chance that the algorithm returns a result that clearly violates privacy, e.g. returning an individual's data record in the data set. Not every approximate DP algorithm satisfies this interpretation, which has resulted in variants of DP, based on R\'enyi divergence.  In particular, adding Gaussian noise with fixed standard deviation $\sigma$ to a statistic of interest cannot be \emph{pure} DP, but can be $(\diffp(\delta),\delta)$-DP for any $\delta>0$.  Hence the probability of failure need not be fixed in advance in the algorithm.  However, there are many different algorithms that are shown to be $(\diffp, \delta)$-differentially private, for a particular $\delta>0$.  In particular, designing DP algorithms for releasing histograms in the Unknown Domain setting will require setting a failure probability in advance.

Consider the example where we want to know the most popular words typed in an email client.  The set of all potential words is massive, and can include slang, typos, and abbreviations, so we only want to take words that are present in the dataset, rather than the set of all possible words.  So one approach would be to add noise to all word counts and sort them to return the top ones.  However, there could be a word like ``RyanRogersSSN123456789".  Even if there is noise in the counts, the mere presence of this record is a clear violation of privacy.  To prevent this, we can introduce a threshold on the noisy counts, but then there is a chance that the noise is especially large, making a small count go above the threshold.  This is where the $\delta>0$ probability comes in, so we want to make sure that the threshold is set high enough to ensure that small counts with noise can only appear above the threshold with probability $\delta$.  However, it does not suffice to just bound the probability of \emph{bad} outcomes to prove an algorithm is approximate differentially private.  

We present a general framework that can be used in the privacy analysis of several algorithms that aims to release counts from an unspecified domain, which would only return results that are present in the original dataset, although the framework can be applied to other scenarios.  The main idea is to show a mechanism $A$ satisfies three conditions: (1) there are small chance events that can lead to \emph{differentiating} outcomes, where it is clear whether one dataset was used, (2) there is another mechanism $A'$ that can know both neighboring datasets, and is equal to $A$ on all non-differentiating outcomes, and (3) $A'$ is pure DP.  Although the algorithms we cover were previously shown to satisfy approximate differential privacy, we revisit their analyses, following our general framework.  Furthermore, we show that providing a privacy analysis in terms of approximate \emph{concentrated DP} (CDP) can lead to improved privacy parameters, especially with algorithms based on Gaussian noise or the Exponential Mechanism.  

We advocate for presenting privacy guarantees of new or existing algorithms in terms of approximate CDP, rather than approximate DP, as composing the privacy parameters for CDP parameters is straightforward, while composing approximate DP parameters can be complicated and loose.  Note that each algorithm is likely to be part of a more general privacy system where the overall privacy guarantee of the system can be in terms of approximate DP.  It has become common to compose algorithms using an analysis based on approximate CDP and then converting to an overall approximate DP guarantee at the end.  Leaving privacy guarantees of individual algorithms in terms of approximate DP will lead to loose bounds in converting the DP parameters to CDP parameters, composing CDP parameters, then lastly converting back to DP at the end.

\section{Preliminaries}

We now define approximate differential privacy which depends on neighboring datasets $x, x' \in \cX$, denoted as $x\sim x'$ that differ in the presence or absence of one user's records.

\begin{definition}[\citet{DworkMcNiSm06, DworkKeMcMiNa06}]
An algorithm  $A : \cX \rightarrow \cY$ is $(\epsilon, \delta)$-differentially private if, for any measurable set $S \subseteq \cY$ and any neighboring inputs $x \sim x'$, 
\begin{equation}
\label{eq:dp}
\Pr[A(x) \in S] \leq e^\epsilon \Pr[A(x') \in S] + \delta.
\end{equation}
If $\delta = 0$, we say $A$ is $\diffp$-DP or simply \emph{pure DP}.  
\end{definition}

One of the classical pure DP algorithms is the Laplace mechanism, which adds Laplace noise to a statistic.  However, to determine the scale of noise to add to the statistic to ensure DP, we must know its \emph{sensitivity}.  We then define the $\ell_p$-sensitivity of a statistic $f:\cX \to \R^d$ that takes a dataset $x \in \cX$ to a real vector in $\R^d$ as the following where the max is taken over neighboring $x, x'\in \cX$
\[
\Delta_p(f) = \max_{x \sim x'} \left\{ ||f(x) - f(x')||_p \right\}.
\]
We then have the following privacy guarantee for the Laplace mechanism.
\begin{theorem}[\citet{DworkMcNiSm06}]
\label{thm:Laplace}
Let $f: \cX \to \R^d$ have $\ell_1$-sensitivity $\Delta_1(f)$, then the mechanism $M: \cX \to \R^d$ where $M(x) = f(x) + (Z_1, \cdots, Z_d)$ with $\{Z_i\}\stackrel{i.i.d.}{\sim} \mathrm{Lap}(\Delta_1(f)/\diffp)$ is $\diffp$-DP for $\diffp>0$.
\end{theorem}

Another classical pure DP mechanism is the Exponential Mechanism, which takes a quality score $q: \cX \times \cY \to \R$ mapping a dataset and outcome to a real value and the goal is to return outcomes that have a high quality score.  We will define the \emph{range} of a quality score to be the following where the max over datasets is over neighbors $x,x' \in \cX$,\footnote{The original Exponential Mechanism was presented in terms of a quality scores \emph{sensitivity}, while recent work from \cite{DongDuRo19} showed that the Exponential Mechanism is more naturally defined in terms of the quality score's range.  See also Jinshuo Dong's blog post \url{https://dongjs.github.io/2020/02/10/ExpMech.html}}
\[
\Delta(q) = \max_{y, y' \in \cY}\max_{x\sim x' } \{ (q(x,y) - q(x', y)) - (q(x,y') - q(x', y'))  \}.
\]
We then have the following privacy guarantee of the Exponential Mechanism.
\begin{theorem}[\citet{McSherryTa07}]
\label{thm:ExpMech}
Let $q: \cX \times \cY \to \R$ be a quality score with sensitivity $\Delta(q)$, then the mechanism $M: \cX \to \cY$ is $\diffp$-DP for any $\diffp>0$ where
\[
\Pr[M(x) = y] \propto \exp\left( \frac{\diffp q(x,y)}{\Delta(q)} \right)
\]
\end{theorem}

We now define approximate concentrated differential privacy (CDP),\footnote{Although \cite{BunSt16} defines zCDP, to differentiate between CDP from \cite{DworkRo16}, we will use CDP to be the version from \cite{BunSt16}} which differs slightly from the original definition but was later shown \cite{WhitehouseRaRoWu22} to be equivalent to the original version.  Similar to approximate DP, it permits a small probability of unbounded R\'enyi divergence. 

\begin{definition}[\citet{BunSt16, PapernotSt22}]\label{def:approxazcdp}
Suppose $A: \cX \to \cY$ and $\rho, \delta\geq 0$. We say the algorithm $A$ is $\delta$-approximate  $\rho$-CDP if, for any neighboring datasets $x, x'$, there exist distributions $P', P'', Q', Q''$ such that the outputs are distributed according to the following mixture distributions:
\begin{align*}
 A(x) \sim (1 - \delta)P' + \delta P''  \qquad A(x') \sim (1 - \delta)Q' + \delta Q'',
\end{align*}
where for all $\lambda\geq 1$, $D_\lambda(P' \| Q') \leq \rho \lambda$ and $D_\lambda(Q' \| P') \leq \rho \lambda$.
\end{definition}

We can also convert approximate differential privacy to approximate CDP and vice versa.
\begin{theorem}[\citet{BunSt16}]
\label{thm:CDPtoDP}
If $A$ is $(\diffp, \delta)$-DP then it is $\delta$-approximate $\diffp^2/2$-CDP.  If $A$ is $\delta$-approximate $\rho$-CDP then it is $(\rho + 2\sqrt{\rho \log(1/\delta')}, \delta' + \delta)$-DP for any $\delta'>0$.  
\end{theorem}

The classical CDP mechanism is the Gaussian Mechanism.  Note that the Gaussian Mechanism was originally introduced as satisfying approximate DP, but it was then shown to satisfy pure CDP in later work \cite{DworkRo16, BunSt16}.  

\begin{theorem}[\citet{BunSt16}]
\label{thm:Gaussian}
Let $f: \cX \to \R^d$ have $\ell_2$-sensitivity $\Delta_2(f)$, then the mechanism $M: \cX \to \R^d$ where $M(x) = f(x) + (Z_1, \cdots, Z_d)$ with $\{Z_i \} \stackrel{i.i.d.}{\sim} \mathrm{N}(0, \tfrac{\Delta_2(f)^2}{2\rho})$ is $\rho$-CDP for $\rho>0$.
\end{theorem}

Note that we can apply Theorem~\ref{thm:CDPtoDP} to conclude that the Exponential Mechanism is $\diffp$-DP and hence $\diffp^2/2$-CDP, but work from \citet{CesarRo20} showed that the Exponential Mechanism actually has a better CDP parameter.\footnote{Also see the previous blog post at \url{https://differentialprivacy.org/exponential-mechanism-bounded-range/}}
\begin{theorem}[\citet{CesarRo20}]
\label{thm:ExpMechBR}
Let $q: \cX \times \cY \to \R$ be a quality score with sensitivity $\Delta(q)$, then the mechanism $M: \cX \to \cY$ is $\diffp^2/8$-CDP for any $\diffp>0$ where
\[
\Pr[M(x) = y] \propto \exp\left( \frac{\diffp q(x,y)}{\Delta(q)} \right)
\]
\end{theorem}

We also state the composition property of CDP, showing that the overall privacy parameters degrade after multiple CDP algorithms are run on a dataset.
\begin{theorem}
\label{thm:compCDP}
Let $A_1: \cX \to \cY$ be $\delta_1$-approximate $\rho_1$-CDP and $A_2: \cX \times \cY \to \cZ$ where $A_2(\cdot, y)$ is $\delta_2$-approximate $\rho_2'$-CDP for all $y \in \cY$.  Then $A: \cX \to \cZ$ where $A(x) = A_2(x, A_1(x))$ is $(\delta_1 + \delta_2 - \delta_1 \cdot \delta_2)$-approximate $\rho_1 + \rho_2$-CDP.
\end{theorem}
\section{Unifying Framework}

We now present a general framework that can be used to unify some of the previous analyses for proving approximate DP (CDP) for various algorithms.  If we want to show an algorithm $A$ is approximate CDP, we need to consider the randomness in $A$ that can generate differentiating outcomes, where if we are given two inputs $x$ and $x'$, we see an outcome of $A$ that could have only come from one of them.  We want to be able to show that the chance that randomness in $A$ can generate these bad outcomes is at most $\delta$.  Furthermore, we want to show that for all other non-differentiating outcomes, there are related mechanisms that match $A$ with inputs $x$ and $x'$.  Lastly, we need to show that these related mechanisms satisfy pure CDP.

To be more precise, the following lemma can be used to prove many different algorithms that are approximate DP can also be proven to be approximate CDP directly, without needing to resort to the general approximate DP conversion to approximate CDP conversion.

\begin{lemma}
\label{lem:unify}
Let $A: \cX \to \cY$ be a randomized algorithm and fix parameters $\rho, \delta \geq 0$.  If for each neighboring datasets $x,x'$ the algorithm $A$ satisfies the following conditions, then it is $\delta$-approximate $\rho$-CDP:

For each neighboring datasets $x,x'$, we have the following three conditions
\begin{itemize}
    \item[1.] There exists events $E, E'$ such that 
    \[
    \Pr_{A(x)}[E], \Pr_{A(x')}[E'] \geq 1-\delta.
    \]
     Let $S$ be the corresponding outcomes of both $A(x)$ conditioned on events $E$ and $A(x')$ conditioned on $E'$.  
    \item[2.] There exists distributions $P'$ and $Q'$ with common support $S$,  such that for all $y \in S$ we have the following where $P$, $Q$ are the distributions for $A(x)$ and $A(x')$, respectively,\footnote{We denote $P(y), P'(y), P''(y), Q(y), Q'(y), Q''(y)$ to denote the Radon-Nikodym derivative of the corresponding distribution with respect to some base measure (see a similar note in \cite{Steinke22})}
    \[
    P(y) = \1{y \in S} \cdot \Pr_{A(x)}[E] P'(y) \qquad Q(y) = \1{y \in S} \cdot\Pr_{A(x')}[E'] Q'(y)
    \]
    \item[3.] Also we have $D_\lambda(P'||Q') \leq \lambda \rho$ and $D_{\lambda}(Q'||P') \leq \lambda \rho$ for all $\lambda \geq 1$.  
\end{itemize}
\end{lemma}
\begin{proof}
    Fix neighbors $x,x'$.  Let $P, Q$ be the distribution for $A(x)$ and $A(x')$, respectively.  Consider the following distribution $P''$ where we write $P(\cdot \mid \neg E)$ to denote the conditional distribution of $P$ given events $E$
    \[
    P''(y) = \frac{1}{\delta} \left( P(y \mid \neg E) \Pr[\neg E]  + P'(y)\left( \Pr_{A(x)}[E] - (1-\delta) \right) \right)
    \]
    We then have that for $y \in S$
    \begin{align*}
    (1-\delta) P'(y) + \delta P''(y) = (1-\delta) P'(y) + \delta \cdot \frac{1}{\delta} \left(P'(y) \cdot (\Pr[E] - (1-\delta))  \right) = P'(y) \Pr[E] = P(y).
    \end{align*}
    Furthermore, for $y \notin S$ we have
    \[
    (1-\delta) P'(y) + \delta P''(y) = \delta \cdot \frac{1}{\delta} \left( P(y \mid \neg E) \Pr[\neg E]  \right) = P(y \mid \neg E) \Pr[\neg E] = P(y).
    \]
    Hence, we have $P(y) = (1-\delta) P'(y) + \delta P''(y)$ for all outcomes $y$.  A similar arguments works to show $Q(y) = (1-\delta) Q'(y) + (1-\delta) Q''(y)$ with $Q''(y)$ defined similar to $P''(y)$.  
    
    By assumption, we know that $D_\lambda(P' || Q') \leq \lambda \rho$ and $D_\lambda(Q' || P') \leq \lambda \rho$ for all $\lambda \geq 1$.  Hence, $A$ is approximate CDP.
\end{proof}

Although a similar lemma can be used for approximate DP, this will typically lead to unnecessarily loose privacy parameters, as we will see later.  Hence, providing an approximate CDP analysis of each algorithm will be useful in any privacy system that combines these algorithms because the CDP privacy parameters simply add up and can be converted to approximate DP parameters at the end, if necessary.   

\section{Unknown Domain Algorithms}

We will denote a histogram as $h \in \cH = \{ (i, c_i) : (i, c_i) \in [d] \times \N \}$ which consists of a set of pairs with a label $i$ in $[d]$ and its corresponding count $c_i$. When we define neighboring histograms, we will need to consider how much a user can modify a histogram.  If we remove a user's contributions from a histogram $h$, this can both remove items entirely or decrease counts by some amount.  We will say that a histogram $h$ is $(\ell_0, \ell_\infty)$-sensitive if removing or adding a user's data to histogram $h$ can change at most $\ell_0$ distinct elements and each count in $h$ can differ by at most $\ell_\infty$.  We now turn to some applications of Lemma~\ref{lem:unify} for some existing mechanisms. 

\subsection{Positive Count Histograms}

We start with the setting where we want to return a histogram subject to CDP where we are only given positive count items, hence each count present in the histogram is at least 1.  It is straightforward to extend this analysis to the case where we have a histogram with only counts above some known value larger than 1.  This is a natural setting for data analytics as GROUP BY queries in SQL only provide items that exist in the dataset, so no zero counts are returned.  This is problematic for privacy, as a neighboring histogram can have fewer results, resulting in a differing set of counts to add noise to.  We present a general template for the private algorithm in Algorithm~\ref{alg:UnkFull}, where we deliberately leave the noise distribution and the threshold arbitrary.

\begin{algorithm}
\caption{Unknown Domain Histogram}\label{alg:UnkFull}
\begin{algorithmic}
    \Require Histogram $h$, noise distribution $\mathrm{Noise}$ and threshold $T>0$
    \Ensure Noisy Histogram $\tilde{h}$ with counts above $T$
    \State Initialize $\tilde{h} = \emptyset$.
    \For{each item $i$ where $(i, c_i) \in h$ such that $c_i >0$}
        \State Set $\tilde{c}_i = c_i+ Z_i$ where $Z_i \sim \mathrm{Noise}$.
        \If{$\tilde{c}_i > T$ }
            \State $\tilde{h} = \tilde{h} \cup \left\{ (i, \tilde{c}_i) \right\}$
        \EndIf
    \EndFor
\end{algorithmic}
\end{algorithm}

Previous mechanisms follow this template, specifically from \citet{KorolovaKeMiNt09} and \citet{WilsonZhLaDeSiGi20} who used Laplace noise and  \citet{SwanbergDeHa23} who used Gaussian noise.  We now prove that Algorithm~\ref{alg:UnkFull} is approximate CDP using Lemma~\ref{lem:unify}.
\begin{theorem}
\label{thm:UnkFull}
    Assume input histograms $h$ are $(\Delta_0, \Delta_\infty)$-sensitive.  If we use $\mathrm{Noise}$ being the distribution of $|h|$ many i.i.d. $\mathrm{Lap}(\Delta_\infty / \diffp)$ and threshold 
    \[
    T = \Delta_\infty + \frac{\Delta_\infty}{\diffp} \log(\tfrac{\Delta_0}{2\delta}),
    \] 
    then Algorithm~\ref{alg:UnkFull} is $\delta$-approximate $\Delta_0\cdot \diffp^2/2$-CDP.
    Furthermore, if we use $\mathrm{Noise} = \mathrm{N}(0, \Delta_\infty^2/\diffp^2\cdot I_{|h|})$ and threshold
    \[
    T = \Delta_\infty + \frac{\Delta_\infty}{\diffp} \Phi^{-1}(1- \delta / \Delta_0),
    \] 
    where $\Phi^{-1}(z)$ is the inverse CDF of a standard normal then Algorithm~\ref{alg:UnkFull} is $\delta$-approximate $\Delta_0 \cdot \diffp^2/2$-CDP.
\end{theorem}
\begin{proof}
    We will rely on Lemma~\ref{lem:unify} to prove this result.  Consider neighbors $h, h'$ where $h$ has one additional user's data from $h'$, without loss of generality.  By assumption, we know that $h$ can differ in at most $\Delta_0$ counts for different labels.  Furthermore, we know that in counts that are differing, they can differ by at most $\Delta_\infty$.  Consider the set $S$ to be the set of labels that are common between $h$ and $h'$ along with corresponding counts.  Since we assume that $h$ has one additional user's data from $h'$, we know that this set $S$ must include all labels of $h'$.  Note that the counts for the items that are not present in $h'$ but are present in $h$ must be at most $\Delta_\infty$.  We now cover each item in  Lemma~\ref{lem:unify}.
    
    \begin{itemize}
    \item[1.]  We define the event $E$ to be all the randomness in $A(h)$ that can generate outcomes in $S$, so that $E$ must only include the noise that is added to items that are common between $h$ and $h'$ and the noise that is added to the items' counts in $h$, but not in $h'$ must be bounded.  We then lower bound the probability of event $E$.  Note that for every item $j$ who is in $h$ just not in $S$, it's count can be no more than $\Delta_\infty$.

    \begin{align*}
        \Pr_{A(h)}[E] &= \prod_{j: (j , \cdot ) \in h \setminus h'} \Pr[c_j + \mathrm{Noise} \leq T] \\
        & \geq \prod_{j: (j , \cdot ) \in h \setminus h'} \Pr[\Delta_\infty + \mathrm{Noise} \leq T]
    \end{align*}
    We then consider the two scenarios with either Laplace noise or Gaussian noise.
    \begin{itemize}
    \item (Laplace) With Laplace noise of scale $b>0$, with $T = \Delta_\infty + b \log(\tfrac{\Delta_0}{2\delta})$, we have
    \begin{align*}
    \Pr_{A(h)}[E] \geq \left(1 - \frac{1}{2} \exp\left( - \frac{T - \Delta_\infty}{b}\right) \right)^{\Delta_0} \geq 1-\delta.
    \end{align*}
    \item (Gaussian) Next we consider the Gaussian noise version that has standard deviation $\sigma >0$ with $T = \Delta_\infty + \sigma \Phi^{-1}(1-\delta/\Delta_0)$.
    \begin{align*}
    \Pr_{A(h)}[E] \geq \Phi\left( \frac{T - \Delta_\infty}{\sigma} \right)^{\Delta_0} \geq 1-\delta.
    \end{align*}
    \end{itemize}
    For this case, the event $E'$ is all randomness in $A(h')$ that can generate outcomes in $S$, which would include all randomness in $A(h')$ because $h'$ is a subset of $h$.  
    
    \item[2.] We then consider the mechanism $A'$ whose domain is the set of common items between $h$ and $h'$ and has noise added to each count so that only noisy counts above $T$ are returned with its corresponding item label.  We write the distribution of $A'(h)$ as $P'$ and the distribution of $A'(h')$ as $Q'$. Because we add independent noise to each histogram count, we can separate out the noise terms added to the counts of labels that are not common between $h$ and $h'$, hence we know that $P(y) = \Pr[E] P'(y)$ for $y \in S$, by design.  Furthermore $Q \equiv Q'$. 
    
    \item[3.] Note that $A'$ is either the Laplace mechanism or the Gaussian mechanism over common items between $h$ and $h'$.  We then cover each variant separately.
    \begin{itemize}
    \item (Laplace) We first consider the case where we apply Laplace noise with noise parameter $b>0$.  Note that $A'$ is a Laplace Mechanism over a histogram that can change in at most $\Delta_0$ counts and each count can differ by at most $\Delta_\infty$. Ignoring the threshold in $A'$, as this is a post processing of the noisy histogram and does not impact the privacy analysis, we can then say that the Laplace mechanism is being applied to a histogram with $\ell_1$-sensitivity $\Delta_0 \Delta_\infty$.  Hence, we know that the Laplace mechanism is $\Delta_0\Delta_\infty/b$-(pure) DP and hence $(\Delta_0^2\Delta_\infty^2/b^2/2)$-(pure) CDP.  However, this will not get us the result we want because it would result in $\Delta_0^2 \diffp^2/2$-CDP with $b = \Delta_\infty/\diffp$.  This was due to using the $\ell_1$-sensitivity of the histogram and then applying Theorem~\ref{thm:Laplace}.  
 
We now consider the Laplace mechanism only on the common items between $h$ and $h'$ that also have differing counts.  We call the corresponding mechanism $\hat{A}$ and denote the distribution $\hat{P}$ for $\hat{A}(h)$ and the distribution $\hat{Q}$ for $\hat{A}(h')$.  Note that each noisy count in $\hat{A}$ is generated independently, so we can say that each count in $\hat{A}$ is a single univariate Laplace mechanism.  Each Laplace mechanism will then be $\Delta_\infty/b$-(pure) DP and hence $(\Delta_\infty/b)^2/2$-(pure) CDP.  Applying composition from Theorem~\ref{thm:compCDP} over each univariate Laplace mechanism implies that $\hat{A}$ is $\frac{\Delta_0(\Delta_\infty)^2}{2b^2}$-(pure) CDP.  Let $\tilde{A}$ denote the Laplace Mechanism applied to the counts that were unchanged between $h$ and $h'$ with distribution $\tilde{P}$ for $\tilde{A}(h) = \tilde{A}(h')$.  For ease of notation, let $h$ and $h'$ match on the first $d'\leq d$ indices and let the first $k \leq \Delta_0$ indices of $h$ and $h'$ have differing counts.  Hence, we have for outcome $y = (y_1, \cdots, y_{d'})$ that $P'(y) = \hat{P}(y_1, \cdots, y_k) \cdot \tilde{P}(y_{k+1}, \cdots, y_{d'})$.  Similarly, we have $Q'(y) = \hat{Q}(y_1, \cdots, y_k) \cdot \tilde{P}(y_{k+1}, \cdots, y_{d'})$.  This gives us the following for $\lambda \geq 1$
\[
D_\lambda(P' || Q') = D_{\lambda} \left( \hat{P} || \hat{Q}  \right) \leq \frac{\Delta_0(\Delta_\infty)^2}{2b^2}
\]
and similarly, 
\[
D_\lambda(Q' || P') \leq \frac{\Delta_0(\Delta_\infty)^2}{2b^2}.
\]


\item (Gaussian) We next consider the Gaussian noise variant.  We will denote $A'$ as the Gaussian mechanism whose domain is the set of common items between $h$ and $h'$ and has noise standard deviation $\Delta_\infty/b$ and only counts above $T$ will be returned with its corresponding item label.  We write the distribution of $A(h)$ as $P'$ and the distribution of $A'(h')$ as $Q'$. 
    
    Note that $A'$ is a Gaussian Mechanism over a histogram that can change in at most $\Delta_0$ counts and each count can differ by at most $\Delta_\infty$. Ignoring the threshold in $A'$, again this does not impact the privacy analysis, we can then say that the Gaussian mechanism is being applied to a histogram with $\ell_2$-sensitivity $\Delta_\infty \sqrt{\Delta_\infty}$.  Hence, we know that the Gaussian mechanism is $\Delta_\infty^2 \Delta_0/\sigma^2$-(pure) CDP from Theorem~\ref{thm:Gaussian}.  Furthermore, post processing the Gaussian Mechanism is also CDP with the same parameters, so restricting the noisy counts to be larger than $T$ gets us back to distributions $P'$ and $Q'$.  This gives us $D_\lambda(P' || Q') \leq \lambda \tfrac{\Delta_\infty^2\Delta_0}{2\sigma^2}$ and  $D_\lambda(Q' || P') \leq \lambda \tfrac{\Delta_\infty^2\Delta_0}{2\sigma^2}$, for all $\lambda \geq 1$.  
\end{itemize}
\end{itemize}

        Hence, using Laplace noise with scale $b>0$ in the Unknown Domain Histogram algorithm is $\delta$-approximate $\frac{\Delta_0^2\Delta_\infty^2}{2b^2}$-CDP.  Setting $b = \Delta_\infty/\diffp$ completes the proof for Laplace noise.  Furthermore, using Gaussian noise with standard deviation $\sigma$ is $\delta$-approximate $\frac{\Delta_\infty^2\Delta_0}{2\sigma^2}$-CDP and setting $\sigma = \frac{\Delta_\infty}{\diffp}$ completes the proof.

\end{proof}

We want to highlight the improvement we get when we use the CDP analysis.  If we consider a similar analysis using approximate DP, we would remove the items that cannot be returned under both neighboring datasets and we would be left with the Laplace mechanism over the common items.  We could then use the $\ell_1$-sensitivity of the resulting histogram, but then adding Laplace noise with scale $b = \Delta_\infty/\diffp$, would result in $(\diffp\Delta_0, \delta)$-DP, which can be converted to CDP using Theorem~\ref{thm:CDPtoDP} to get $\delta$-approximate $\diffp^2\Delta_0^2/2$-CDP, although we can get $\delta$-approximate $\diffp^2\Delta_0/2$-CDP in our analysis.  Furthermore, if we convert approximate CDP guarantees to approximate DP guarantees, this would lead to loose privacy parameters when developing privacy systems that use these algorithms.  For example, if we use the Gaussian noise variant in Theorem~\ref{thm:UnkFull} with $\Delta_0 = 1$,  we can conclude that it is $(\diffp^2 / 2 + \diffp \sqrt{2 \log(1/\delta) },\delta + \delta')$-DP for any $\delta'>0$, so if we only use DP guarantee and combine it with another Unknown Domain Histogram with Gaussian noise, we can compose to get $(\diffp', 2\delta + \delta')$-DP for any $\delta'>0$ where
\[
\diffp' = \diffp^2 + 2\diffp \sqrt{2 \log(1/\delta') }.
\]
However, if we had never converted to approximate DP until after composing both Unknown Domain Histogram mechanisms, we could have gotten an overall $(\diffp'', 2\delta + \delta')$-DP guarantee, where
\[
\diffp'' = \diffp^2 + 2\diffp \sqrt{ \log(1/\delta') }.
\]

\subsection{Top-$(\bar{k}+1)$ Count Histograms}

We now turn to a slight variant of releasing private histograms over a histogram with positive counts.  In this setting, we assume that only a limited part of the histogram is available, perhaps due to an existing data analytics system that cannot return all counts.  This setting first appeared in \cite{DurfeeRo19} and is especially important when designing DP algorithms on top of existing systems that cannot provide the full histogram of counts \cite{RogersLinkedIn21}.  We will refer to the histogram consisting of the top-$(\bar{k}+1)$ items as the histogram with items and corresponding counts that are in the top-$(\bar{k}+1)$.  Note that in this case, the top-$(\bar{k}+1)$ can change between neighboring histograms.  We now present a general algorithm template, similar to Algorithm~\ref{alg:UnkFull}, in Algorithm~\ref{alg:UnkTop} that takes an arbitrary threshold $T>0$ and a top-$(\bar{k}+1)$-histogram.

\begin{algorithm}
\caption{Unknown Domain from Top-$(\bar{k}+1)$}\label{alg:UnkTop}
\begin{algorithmic}
    \Require Histogram $h$, noise standard deviation $\sigma$, threshold $T>0$, and top-$(\bar{k}+1)$ histogram 
    \Ensure Noisy Histogram $\tilde{h}$ with at most $\bar{k}$ counts.
    \State Let $h_{(\bar{k}+1)}$ be the histogram consisting of the top-$(\bar{k}+1)$ items, breaking ties arbitrarily.
    \If{$h_{(\bar{k})}$ has fewer than $\bar{k}$ items}
    	\State Pad $h_{(\bar{k})}$ with items $\bot_1, \cdots, \bot_{\bar{k} - |h_{(\bar{k})}|}$ with $c_{\bot_j} = 0$ until there are $\bar{k}$ items in $h_{(\bar{k})}$.
    \EndIf
    \State Let $c_{(\bar{k}+1)}$ be the count of the $(\bar{k} + 1)$-th item in $h$, which might be zero.
    \State Set $\tilde{T} = T + c_{(\bar{k}+1)}+\mathrm{N}(0,\sigma^2)$
    \State Initialize $\tilde{h} = \emptyset$.
    \For{each item $i$ where $(i, c_i) \in h_{(\bar{k})}$}
        \State Set $\tilde{c}_i = c_i+ \mathrm{N}(0,\sigma^2)$
        \If{$\tilde{c}_i > \tilde{T}$ }
            \State $\tilde{h} = \tilde{h} \cup \left\{ (i, \tilde{c}_i) \right\}$
        \EndIf
    \EndFor
\end{algorithmic}
\end{algorithm}

We now show that it is indeed approximate CDP.

\begin{theorem}
    \label{thm:UnkTop}
    Assume input histograms $h$ are $(\Delta_0, \Delta_\infty)$-sensitive.  
    If we use $\sigma =  \frac{\Delta_\infty}{\diffp}$
    and threshold 
    \[
    T =  \Delta_\infty + \frac{\sqrt{2} \cdot \Delta_\infty }{\diffp} \Phi^{-1}(1- \delta / \Delta_0)
    \]
    then Algorithm~\ref{alg:UnkTop} is $\delta$-approximate $\Delta_0\cdot \diffp^2/2$-CDP. 
\end{theorem}

\begin{proof}
    We follow the same analysis as in the proof of Theorem~\ref{thm:UnkFull}, which used Lemma~\ref{lem:unify}.  We again set $S$ to be the set of labels that are common between neighbors $h$ and $h'$.  Note that we are only considering items with counts in the top-$(\bar{k})$ of each histogram and the items between $h_{(\bar{k})}$ and $h_{(\bar{k})}'$, including the zero count items with labels in $\{ \bot_j \}$.  Hence, there might be different items between $h_{(\bar{k})}$ and $h_{(\bar{k})}$, as reducing some counts in $h$ might change the order of the top-$(\bar{k}+1)$. From Lemma 5.2 in \cite{DurfeeRo19}, we know that there can be at most $\min\{\bar{k}, \Delta_0\}$ many differing items between the top-$(\bar{k}+1)$ histograms  $h_{(\bar{k})}$ and $h_{(\bar{k})}'$.  We then follow the three items that we need to show in order to apply Lemma~\ref{lem:unify}.
    
    \begin{itemize}
    
    \item[1.] We denote $A$ as Algorithm~\ref{alg:UnkTop} and the event $E$ as the noise added to counts in outcomes $S$ for $A(h)$, the noise added for the threshold $T + c_{(\bar{k}+1)}$ to get $\tilde{T}$, and the noise added to the differing items not in $S$ must be no more that $\tilde{T}$.  We define $E'$ similarly for $A(h')$.  Note that for any item $j$ that is in $h_{(\bar{k})}$ but not an item that can be returned in $S$, we know
    \[
    c_j \leq c_j' + \Delta_\infty \leq c'_{(\bar{k})} + \Delta_\infty \leq  c_{(\bar{k})} + \Delta_\infty \leq  c_{(\bar{k}+1)} + \Delta_\infty.
    \]
    The analysis is straightforward due to the difference between two Gaussians being Gaussian itself.  Hence, we have with $\hat{T} = \sqrt{2} \sigma \Phi^{-1}(1-\delta/\Delta_0)$ so that $T = \hat{T} + \Delta_\infty$.

    \begin{align*}
        \Pr_{A(x)}[\neg E] & \leq \sum_{i = 1}^{\Delta_0} \Pr \left[ c_{(\bar{k}+1)}  + \Delta_\infty +  \mathrm{N}(0, \sigma^2) >  c_{(\bar{k}+1)} +T+ \mathrm{N}(0, \sigma^2) \right] \\
        & = \Delta_0 \Pr[ \mathrm{N}(0, 2 \sigma^2) >   \hat{T}] \\
        & = \Delta_0 \left( 1 - \Phi\left(\frac{\hat{T}}{\sqrt{2}\sigma} \right) \right) \\
        & = \Delta_0 \left( 1 - (1-\delta/\Delta_0) \right) = \delta
    \end{align*}
    
    The analysis to show $ \Pr_{A(x')}[\neg E'] $ is similar.
    
    \item[2.] We then consider the distribution $P'(y)$ to be the distribution of $A(h)$ conditioned on $E$.  Note that $P(y) = \Pr[E] P'(y)$ for each $y \in S$ and similarly $Q(y) = \Pr[E'] Q'(y)$ for $Q'(y)$ being the distribution of $A(h')$ conditioned on $E'$.  We will modify the mechanism $A(h)$ and $A(h')$ that will result in the same mechanism when conditioned on events $E$ and $E'$, respectively.  For each label that differs between $h_{(\bar{k})}$ and $h'_{(\bar{k})}$, we change it to common labels $b_1, \cdots b_{\ell}$ where $\ell = |\{ j : (j, \cdot) \in \{ h_{(\bar{k})} \setminus h'_{(\bar{k})} \}   \}| \leq \Delta_0$.  Because we condition on events $E$, no outcome $y \in S$ can include the indices $b_1, \cdots, b_\ell$.  Furthermore, it was shown in \cite{CardosoRo22}[Lemma 6.4]] that the resulting histogram with common labels will also have as many as $\Delta_0$ items with differing counts, and those counts can change by at most $\Delta_\infty$, regardless of how we assign the common labels $\{b_j\}$.  We then let $P'$ and $Q'$ be the resulting distribution after this relabeling.
    
    \end{itemize}
    Next we will need to bound the R\'enyi divergence between $P'$ and $Q'$.   We will make use of the following result from \cite{GopiLeLi22}[Corollary 4.3]
\begin{lemma}
\label{lem:convexCDP}
Suppose $F, F'$ are two $\mu$-strongly convex functions over $\cK \subseteq \R^d$, and $F - F'$ is $G$-Lipschitz over $\cK$.  For any $k >0$, if we let $P \propto e^{-mF}$ and $Q \propto e^{-m F"}$ be two probability distributions on $\cK$, then we have for all $\lambda \geq 1$
\[
D_\lambda(P||Q) \leq \frac{\lambda m G^2}{2\mu}
\]
\end{lemma}

This result is useful because it allows us to condition on outcomes from a joint Gaussian Mechanism falling in some convex region, which will correspond to releasing only ``good" outcomes, i.e. not allowing certain counts from going above some noisy value.  For the Gaussian mechanism, we have $F(z) = || z - h||_2^2$ and $F'(z) = ||z - h'||_2^2$, which are both $2$-strongly convex over any convex region.  Furthermore, we have $||z - h||_2^2 - || z - h'||_2^2$ is $2\sqrt{\Delta_0} \cdot \Delta_\infty$-Lipschitz.  For the density of a Gaussian, we then use $m = \tfrac{1}{2\sigma^2}$

    \begin{itemize}
    
    \item[3.] We now want to prove that the R\'enyi divergence between $P'$ and $Q'$, which are conditioned on events in $E$ and $E'$ respectively.  To do this we will consider the joint Gaussian distribution that releases all counts over the histograms with common labels, including the $(\bar{k} + 1)$-th largest count with added $T$ to its count being labeled $\bot$, but we do not enforce the threshold.  We only want to consider events that do not have the items with labels in $\{ b_j\}$ having noisy count above the noisy count for $\bot$.  We then consider the convex region $\cK$ where these ``bad" noisy counts do not go above the noisy count for $\bot$.  We then apply Lemma~\ref{lem:convexCDP} to claim that the resulting mechanism conditioned on this region has a bound on the R\'enyi Divergence that is the same as if it were the Gaussian mechanism not constrained to region $\cK$.  Dropping the items with counts lower than the noisy count for $\bot$ is simply post processing, which does not increase the R\'enyi divergence bound.  Hence, we have $D_\lambda(P'||Q') \leq \frac{\lambda \Delta_0 \Delta_\infty^2}{2\sigma^2} $ and  $D_\lambda(P'||Q') \leq \frac{\lambda \Delta_0 \Delta_\infty^2}{2\sigma^2}$ for all $\lambda \geq 1$. 
    
    \end{itemize}
    Setting $\sigma = \Delta_\infty/\diffp$ completes the proof.
\end{proof}

\subsubsection{Exponential Mechanism}

For the previous applications, there was a crucial assumption that the input histograms were $(\Delta_0, \Delta_\infty)$-sensitive, specifically that the histogram's $\ell_0$-sensitivity must be bounded by $\Delta_0$.  However, this might not always be the case, and would be difficult to enforce a bound in practice -- requiring that each user has only a certain number of distinct items in the data.  One way to limit the impact a single user can have on the result is to limit the number of items that can be returned.  The Laplace/Gaussian Mechanism based algorithms can return an arbitrary number of things that are above the threshold and in the setting where we only have access to the top-$(\bar{k}+1)$ items, we could return all $\bar{k}$.  Hence, the CDP parameter could be bounded in terms of $\bar{k}$, but we would like to control how many things could be returned from the counts we have access to.  In this case, we can return the top-$k$ results, where $k \leq \bar{k}$ is an input to the algorithm.  Unfortunately, the previous analysis would still have the CDP parameter scale with $\bar{k}$ despite only wanting to return $k$.  

To ensure that privacy loss only scales with the number of things that are returned, we can use the classical Exponential Mechanism \cite{McSherryTa07}, as presented in Theorem~\ref{thm:ExpMech}.  It was shown that the Exponential Mechanism is part of a general framework of DP mechanisms called \emph{report noisy max}.  This can be summarized as adding noise to the quality scores for each outcome and returning the outcome with the noisy max.  Note that it is critical that only the arg max is returned, not the actual noisy value.  It turns out that adding Gumbel noise to the quality scores and returning the max item is equivalent to the Exponential Mechanism, see \cite{DurfeeRo19}.  Furthermore, it was shown that iteratively applying the Exponential Mechanism to return the top-$k$ outcomes is equivalent to adding Gumbel noise to all quality scores and returning the outcomes with the $k$ largest noisy quality scores \cite{DurfeeRo19}.  Other mechanisms in the report noisy max framework include adding Laplace noise to the quality scores \cite{DworkRo14} and adding Exponential noise to the quality scores \cite{DingKiSaStWaXiZh21} which turns out to be equivalent to the Permute-and-Flip Mechanism \cite{McKennaSh20}.\footnote{See previous blog post on \emph{one-shot} top-$k$ DP algorithms: \url{https://differentialprivacy.org/one-shot-top-k/}}

We then present the Unknown Domain Gumbel algorithm in Algorithm~\ref{alg:UnkGumbel} from \cite{DurfeeRo19} which takes an additional parameter $k \leq \bar{k}$ and importantly returns a ranked list of items, but not their counts.  
\begin{algorithm}
\caption{Unknown Domain Gumbel from Top-$(\bar{k}+1)$}\label{alg:UnkGumbel}
\begin{algorithmic}
    \Require Histogram $h$, noise scale $\beta>0$, threshold $T>0$, and $k, \bar{k}$
    \Ensure Sorted list of at most $k$ items, $I_k$.
    \State Let $h_{(\bar{k})}$ be the histogram consisting of the top-$(\bar{k})$ items 
    \State Let $c_{(\bar{k}+1)}$ be the count of the $(\bar{k} + 1)$-th item in $h$.
    \State Set $\tilde{T} = T + c_{(\bar{k}+1)}+\mathrm{Gumb}(\beta)$
    \State Initialize $\tilde{h} = \emptyset$.
    \For{each item $i$ where $(i, c_i) \in h_{(\bar{k})}$ such that $c_i >0$}
        \State Set $\tilde{c}_i = c_i+ \mathrm{Gumb}(\beta)$
        \If{$\tilde{c}_i > \tilde{T}$ }
            \State $\tilde{h} = \tilde{h} \cup \left\{ (i, \tilde{c}_i) \right\}$
        \EndIf
    \EndFor
    \State Let $I_k$ be the ordered list of at most $k$ items that are sorted in descending order by their count in $\tilde{h}$.  
    \If{$I_k<k$}
        \State $I_k = I_k \cup \{ \bot \}$
    \EndIf
\end{algorithmic}
\end{algorithm}
We then state Unknown Domain Gumbel's privacy guarantee, which will largely follow the analysis in \cite{DurfeeRo19}, although adapted for CDP.
\begin{theorem}
\label{thm:unkGumbel}
    Assume input histograms $h$ are $(\infty, 1)$-sensitive.  If we use the noise scale $\beta = 1/\diffp$, and threshold 
    \[
    T =  1 + \frac{1 }{\diffp} \log(\tfrac{\bar{k}}{\delta}),
    \] 
    then Algorithm~\ref{alg:UnkGumbel} is $\delta$-approximate $k\diffp^2/8$-CDP.
\end{theorem}
\begin{proof}
    We will apply our general framework from Lemma~\ref{lem:unify} and some previous results from \cite{DurfeeRo19}.  We will denote $A$ as the Unknown Domain Gumbel mechanism.  We assume WLOG that $h, h'$ differ by there being one additional user's data in $h$ when compared to $h'$. 
    \begin{itemize}
    \item[1.] We denote $S$ as the set of common outcomes that are possible between top-$(\bar{k})$ histograms $h_{(\bar{k})}$ and $h'_{(\bar{k})}$ and $E$ to be all the randomness in $A(h)$ that can generate outcomes in $S$.  Hence, $E$ must only include the Gumbel noise added to items that are common in $h$ and $h'$ and the noise that is added to the differing counts must have a noisy count either below the top-$k$ or below the noisy threshold $\tilde{T}$ that is set.  From Lemma 5.5 in \cite{DurfeeRo19}, we have the following bound when $T = 1 + \beta \log\left(\Delta_0/\delta \right)$
    \[
    \Pr_{A(h)}[E] \geq 1-\delta.
    \]
    We similarly define events $E'$ for the randomness in $A(h')$ that can generate outcomes in $S$ which gives us the same lower bound for $\Pr_{A(h')}[E']$.
    \item[2.] From Lemma 5.6 in \cite{DurfeeRo19}, we know that there exists a distribution $P', Q'$ such that for all outcomes $y \in S$, we have
    \[
    \Pr[A(x) = y] = \Pr[E] P'(y) \qquad \Pr[A(x') = y] = \Pr[E'] Q'(y).
    \]
    \item[3.] These distributions  $P', Q'$ are also related to the Exponential Mechanism.  Specifically, $P'$ is the distribution of iteratively applying the Exponential Mechanism on $h_{(\bar{k})}$ only over the items that are common between $h_{(\bar{k})}$ and $h_{(\bar{k})}$.  Similarly, we can define $Q'$ as the distribution of iteratively applying the Exponential Mechanism on $h_{(\bar{k})}'$ from items that are common between $h_{(\bar{k})}$ and $h_{(\bar{k})}$.  We can then use Theorem~\ref{thm:ExpMechBR} to conclude that $D_\lambda(P'||Q') \leq \frac{k}{8\beta^2}$ and $D_\lambda(Q'||P') \leq \frac{k}{8\beta^2}$, for all $\lambda\geq 1$. 
    \end{itemize}
     Setting $\beta = \frac{1}{\diffp}$ completes the proof.
\end{proof}

\subsubsection{Pay-what-you-get Composition}
We now consider the setting where an analyst wants to run multiple top-$k$ queries over an unknown domain.  One of the issues with the Unknown Domain Gumbel mechanism in Algorithm~\ref{alg:UnkGumbel} is that it can sometimes return fewer than $k$ results, yet the overall CDP privacy parameter scales with $k$, regardless of the number of results returned.  Hence, with $\ell^*$ many top-$k$ queries, the privacy loss would scale with $\ell^* \cdot k$ with traditional privacy techniques.  However, \cite{DurfeeRo19} showed that if we instead set an overall bound on the number of results that can be returned $k^*$ from all the intermediate top-$k$ queries, then the privacy loss will scale with $k^*$.  They referred to this as \emph{pay-what-you-get} composition because even if an analyst requests top-100 items, and only 1 result is returned, we need only deduct the number of items that are returned (including $\bot$) from the overall privacy budget until $k^*$ items have been returned.  We present the pay-what-you-get composition in Algorithm~\ref{alg:paywhatyouget}

\begin{algorithm}
\caption{Pay-what-you-get Mechanism}\label{alg:paywhatyouget}
\begin{algorithmic}
    \Require Sequence of histogram $h^{(1)}, \cdots, h^{(\ell^*)}$, noise scale $\beta>0$, threshold $T>0$, and $k^*$
    \Ensure Sequence of sorted lists $\{I^{(i)} : i \in [\ell^*]\}$.
    \State Initialize $k = 0$
    \For{$i \in [\ell^*]$}
        \State Apply the Unknown Domain Gumbel from top-$(\bar{k}^{(i)}+1)$ on histogram $h^{(i)}$ with noise scale $\beta$, threshold $T>0$, along with $k^{(i)} \leq (k^* - k)$ and $ \bar{k}^{(i)}$ to get items $I^{(i)}$, which includes $\bot$ in some cases.
        \State Update $k = k + |I^{(i)}|$
    \If{$k = = k^*$}
        \State \bf{break}
    \EndIf
    \EndFor
\end{algorithmic}
\end{algorithm}

The key observation to proving the overall privacy guarantee of pay-what-you-get is that each Unknown Domain Gumbel Mechanism is related to the Exponential Mechanism.  Removing the randomness that can generate outcomes from the differing items between neighboring datasets, the algorithm becomes equivalent to an Exponential Mechanism.  Hence, considering only the randomness over each Unknown Domain Gumbel that can generate outcomes in both neighboring histograms, we are left with just a string of at most $k^*$ many Exponential Mechanisms.  

\begin{theorem}
    Given an adaptive stream of histograms $h^{(1)}, \cdots, h^{(\ell^*)}$, where each histogram $h^{(i)}$ is $(\infty, 1)$-sensitive, the Pay-what-you-get mechanism in Algorithm~\ref{alg:paywhatyouget} with $\beta = \Delta_\infty/\diffp$ and threshold $T$ being the same as in Theorem~\ref{thm:unkGumbel} with $\delta>0$ is $\ell^*\delta$-approximate $k^* \diffp^2/8$-CDP.
\end{theorem}
\begin{proof}
    We let $A$ be the Pay-what-you-get mechanism with neighboring histogram streams $(h^{(1)}, \cdots, h^{(\ell^*)})$ and $(h'^{(1)}, \cdots, h'^{(\ell^*)} )$, i.e. each histogram $h^{(i)}$ and $h'^{(i)}$ are neighbors, we denote $P$ as the distribution for $A(h^{(1)}, \cdots, h^{(\ell^*)})$ and $Q$ as the distribution for $A(h'^{(1)}, \cdots, h'^{(\ell^*)} )$.  WLOG, we assume that each $h^{(i)}$ has an additional user's data from $h'^{(i)}$ for $i \in [\ell^*]$.  

    We will denote $S = (S^{(1)}, \cdots, S^{(\ell^*)})$ as the set of common outcomes between these neighboring streams.  
    
    \begin{itemize}
    \item[1.] Let $E = (E^{(1)}, \cdots, E^{(\ell^*)})$ be the corresponding randomness in $A(h^{(1)}, \cdots, h^{(\ell^*)}) = A^{(1)}(h^{(1)}), \cdots, A^{(\ell^*)}(h^{(\ell^*)}),$ that can generate outcomes in $S = (S^{(1)}, \cdots, S^{(\ell^*)})$, where each $A^{(i)}$ is an Unknown Domain Gumbel mechanism, and similarly we denote $E' = (E'^{(1)}, \cdots, E'^{(\ell^*)})$ as the corresponding randomness in $A(h'^{(1)}, \cdots, h'^{(\ell^*)})$ that can generate outcomes in $S$.   
    We then apply Lemma 5.5 in \cite{DurfeeRo19} to get the following for each event set of $A^{(i)}$
    \[
        \Pr_{A^{(i)}(h^{(i)})}[E^{(i)} ] \geq 1- \delta, \qquad \Pr_{A^{(i)}(h'^{(i)})}[E'^{(i)}] \geq 1- \delta, \qquad \forall i \in [\ell^*].
    \]
    Hence, applying a union bound over all $\ell^*$ events, we get 
    \[
    \Pr[E] \geq 1- \ell^*\delta, \qquad \Pr[E'] \geq 1- \ell^* \delta.
    \]
    
    \item[2.] We then apply Lemma 5.6 in \cite{DurfeeRo19}, as we did in Theorem~\ref{thm:unkGumbel} to write the distribution of $A$ in terms of a stream of top-$k^{(i)}$ mechanisms that can be further decomposed into Exponential Mechanisms.  Given the prior outcomes $y^{(<i)}$, we will write $P'^{(i)}(\cdot; y^{(<i)})$ to denote the distribution of the top-$k^{(i)}$ mechanism evaluated on $h^{(i)}_{(\bar{k}^{(i)})}$ but only on the common items between $h^{(i)}_{(\bar{k}^{(i)})}$ and $h'^{(i)}_{(\bar{k}^{(i)})}$.  Similarly, we will write $Q'^{(i)}(\cdot; y^{(<i)})$ to denote the distribution of the top-$k^{(i)}$ mechanism evaluated on $h'^{(i)}_{(\bar{k}^{(i)})}$ but only on the common items with $h^{(i)}_{(\bar{k}^{(i)})}$.  We can now apply Lemma 5.6 from \cite{DurfeeRo19} iteratively to get for all $y = (y^{(1)}, \cdots, y^{(\ell^*)})\in S$,
    \begin{align*}
    \Pr[A(h^{(1)}, \cdots, h^{(\ell^*)})=y] &= \prod_{i=1}^{\ell^*} \Pr[E^{(i)}] P'^{(i)}(y^{(i)}; y^{(<i)}) \\
    \Pr[A(h'^{(1)}, \cdots, h'^{(\ell^*)})=y] &= \prod_{i=1}^{\ell^*} \Pr[E'^{(i)}] Q'^{(i)}(y^{(i)}; y^{(<i)}).
    \end{align*}

    \item[3.] We then consider the distribution $P'(y) \defeq \prod_{i=1}^{\ell^*}P'^{(i)}(y^{(i)}; y^{(<i)})$, which can be further decomposed into separate Exponential Mechanisms.  We will write $y^{(i)} = (y^{(i)}[1], \cdots, y^{(i)}[|y^{(i)}|])$ and $\hat{P}^{(i,j)}$ as the distribution for the Exponential Mechanism for the $i$-th round's mechanism returning the $j$-th item.  Hence we have
    \[
    P'(y)=\prod_{i=1}^{\ell^*}P'^{(i)}(y^{(i)}; y^{(<i)}) = \prod_{j=1}^{\ell^*} \prod_{j = 1}^{|y^{(i)}|} \hat{P}^{(i,j)}(y^{(i)}[j]; y^{(<i)}, y^{(i)}[1], \cdots, y^{(i)}[j-1]).  
    \]
    Note that the two products on the right hand side will have at most $k^*$ many terms due to the number of items that can be returned by pay-what-you-get, and each item is the distribution of an Exponential Mechanism.  We have a similar argument for $Q'^{(i)}(\cdot; y^{(<i)})$ being decomposed as a sequence of Exponential Mechanisms with $ Q'(y)=\prod_{i=1}^{\ell^*}Q'^{(i)}(y^{(i)}; y^{(<i)}) $.    Hence, we have 
    $D_{\lambda}\left( P' ||  Q' \right) \leq \lambda k^*\diffp^2/8$ and $D_{\lambda}\left( Q' || P' \right) \leq \lambda k^*\diffp^2/8$ for all $\lambda \geq 1$.  
    \end{itemize}
    Note that we can almost apply Lemma~\ref{lem:unify}, but not quite because we have $P(y) = \Pr[A(h^{(1)}, \cdots, h^{(\ell^*)})=y] = \prod_{i=1}^{\ell^*} \Pr[E^{(i)}] P'^{(i)}(y^{(i)}; y^{(<i)}) $ when $y \in S$.  We then construct the joint distribution as 
    \begin{align*}
    P(y^{(1)}, \cdots, y^{(\ell^*)}) &= \prod_{i =1}^{\ell^*} P_i(y^{(i)} | y^{(<i)}) =  \prod_{i =1}^{\ell^*} \left( \Pr[E^{(i)}]  P_i'(y^{(i)} | y^{(<i)}) + \Pr[\neg E^{(i)}] P_i(y^{(i)} | y^{(<i)},\neg E^{(i)} )\right) \\
    & =  \prod_{i=1}^{\ell^*} \Pr[E^{(i)}] P'^{(i)}(y^{(i)}; y^{(<i)}) \\
    & \quad + \sum_{S\subseteq [\ell^*] \setminus \emptyset} \left( \prod_{i \in S}\Pr[E^{(i)}] \prod_{i \notin S}  \Pr[\neg E^{(i)}]  \right) \left( \prod_{i \in S} P_i'(y^{(i)} | y^{(<i)} \prod_{i \notin S} P_i(y^{(i)} | y^{(<i)},\neg E^{(i)} )\right) \\
   &  \eqdef  P'(y^{(1)}, \cdots, y^{(\ell^*)}) \cdot  \prod_{i=1}^{\ell^*} \Pr[E^{(i)}] + \hat{P}''(y^{(1)}, \cdots, y^{(\ell^*)})
    \end{align*}

    We can then write the probability as a convex combination, so that for all $y$ we have the following
    \begin{align*}
    \Pr[A(h^{(1)}, \cdots, h^{(\ell^*)})=y] & =P'(y^{(1)}, \cdots, y^{(\ell^*)}) \cdot  \prod_{i=1}^{\ell^*} \Pr[E^{(i)}]  +  \hat{P}''(y)
    \\
    \implies \Pr[A(h^{(1)}, \cdots, h^{(\ell^*)})=y] & = (1-\delta) P'(y) \\
    & \quad + \delta \underbrace{\left( 1/\delta \cdot \left( \hat{P}''(y)+  P'(y)\left(  \prod_{i=1}^{\ell^*} \Pr[E^{(i)}]  - (1-\delta) \right)  \right)   \right)}_{P''(y)}
    \end{align*}
    Similarly, we have 
    \[
     \Pr[A(h'^{(1)}, \cdots, h'^{(\ell^*)})=y] = (1-\delta)\underbrace{\prod_{i=1}^{\ell^*}Q'^{(i)}(y^{(i)}; y^{(<i)})}_{Q'(y)} + \delta Q''(y).
    \]
    This completes the proof.
\end{proof}

\subsection{Continual Observation}
We now present an approach for continually releasing a running counter over various domain elements while ensuring differential privacy.  The \emph{continual observation} privacy model was introduced in \cite{ChanShSo12, DworkNaPiRo10} and is meant to ensure strong privacy guarantees in the setting where we want to continually provide a running counter on a stream of events, denoted at $\omega^{(1:\ell)} = \left( \omega^{(1)}, \cdots, \omega^{(\ell)} \right)$ for $\ell = 1, \cdots, L$ where $\omega^{(i)} \in \{0,1 \}$.  It is then straightforward to extend the continual observation counters to the setting of releasing a running histogram over a known set of items, i.e. $\omega^{(i)} \subseteq [d]$, where someone can contribute a limited set of items $|\omega^{(i)}| \leq \Delta_0$ at each round $i$ in the stream.  Typically we want to ensure \emph{event-level} privacy, where we consider the change in outcomes when one event in the stream can change.

Recent work from \cite{CardosoRo22, ZhangVaKaStThMaApSp23} have also considered the continual observation setting, but in the case where we want to continually release histograms over an unknown set of items.  Consider the motivating example where we want to provide a running counter for the drugs that are purchased at a pharmacy and each customer can buy at most $\Delta_0$ different drugs.  There might be several different drugs to provide a counter for and new drugs can emerge later that were not known about before.  Hence, we would like to have algorithms that do not require a set of items to provide counts over.

We describe the algorithm at a high level, as formally describing will require some additional notation.  The main subroutine is the Binary Mechanism from \cite{ChanShSo12, DworkNaPiRo10} which will add at most $\log_2(\ell)$ many noise terms to the count of a stream of length $\ell$ events.  The number of noise terms depends on the number of 1s in the binary representation of $\ell$.  This will ensure that for a stream of length at most $L$, we can release $L$ counts, one after each event, each with Gaussian noise with standard deviation $O(\log_2(L)/\diffp)$.  Although we release $L$ counts, the privacy analysis relies on the fact that we form a table of \emph{partial sums}, so that one event in the stream can modify at most $ \log_2(L)$ partial sums and we can view the Binary Mechanism as a Gaussian mechanism on the table of partial sums for all common items between $\omega^{(1:L)}$ and $\omega'^{(1:L)}$, which has $\ell_2$-sensitivity $\Delta_0 \cdot \log_2(L)$.

The Unknown Domain Binary Mechanism follows the same approach as the classical Binary Mechanism, which we present at a high level.  The idea on the Binary Mechanism is to split a stream of items $\omega^{(1:\ell)}$  into overlapping partial sums, guaranteeing that each $\omega^{(i)}$ is part of no more than $\lceil \log_2(L+1) \rceil$ partial sums.  We create separate partial sums for each item in $\omega^{(1:\ell)}$.  For instance, with $\ell = 8$, we will focus on a single partial sum, based on the binary representation of $\ell$, so that we need only add noise to the partial sums $\sum_{j = 1}^8 \1{u \in  \omega^{(j)}}$ for each $u \in \omega^{(1:\ell)}$ which we add $\mathrm{N}(0, \sigma^2)$ noise to the partial sum and is then used again for any other partial sum that utilizes  $\sum_{j = 1}^8 \1{u \in  \omega^{(j)}}$.  For instance, when $\ell = 10$, we will use the two partial sums $\sum_{j = 1}^8 \1{u \in  \omega^{(j)}}$ and $\sum_{j = 9}^{10} \1{u \in  \omega^{(j)}}$ for each $u \in \omega^{(1:10)}$, each with its own noise added to it.  Note that for each $u \in \omega^{(1:\ell)}\setminus \omega^{(1:\ell-1)}$, we will add fresh noise to each prior partial sum for the new item $u$.  We then only release items with corresponding noisy counts if it is larger than a fixed threshold $T>0$.

We now present the privacy analysis for Unknown Domain Binary Mechanism from \cite{CardosoRo22}.

\begin{theorem}
    Assume that $|\omega^{(\ell)}| \leq \Delta_0$ for all $ \ell \in [L]$.  Setting $\sigma = 1/\diffp$ and threshold $T$ to be the following for any $\delta>0$
    \[
    T = 1 + \sigma \cdot \sqrt{\lceil \log_2(L+1) \rceil + 1} \cdot \Phi^{-1}\left(1-\frac{\delta}{\Delta_0\cdot L}\right)
    \]
    ensures that Unknown Domain Binary Mechanism is $\delta$-approximate $\Delta_0 \cdot \lceil \log_2(L+1) \rceil\diffp^2/2$-CDP, under event-level adjacent streams.
\end{theorem}
\begin{proof}
    We follow the same analysis as in the earlier theorems, mainly leveraging Lemma~\ref{lem:unify}.  Let $\omega^{(1:L)}$ contain an event where neighboring stream $\omega'^{(1:L)}$ has an empty set at that event.  Say that round $\ell$ is where they differ so that $\omega'^{(\ell)} = \emptyset$ and $|\omega^{(\ell)}| = \Delta_0$.  Let $A$ denote the Unknown Domain Binary Mechanism.  We denote $S$ as the set of all outcomes that both $A(\omega^{(1:L)})$ and $A(\omega'^{(1:L)}$ can return.  Note that at round $\ell$, stream event $\omega^{(\ell)}$ can introduce as many as $\Delta_0$ previously unseen items in the stream.  We will write the distribution of $A(\omega^{(1:L)})$ as $P$ and the distribution of $A(\omega'^{(1:L)})$ as $Q$.  
    \begin{itemize}
    \item[1.]     We will write $E$ as the randomness in $A(\omega^{(1:L)})$ that can generate outcomes that are common between $A(\omega^{(1:L)})$ and $A(\omega'^{(1:L)})$.  We need to ensure that no noisy count on any new item that appears in $\omega^{(\ell)}$ but not in $\omega^{(1:\ell-1)}$ can go above the threshold $T$.  At round $\ell$ we can add together as many as $(\lceil\log_2(\ell+1) \rceil ) \leq (\lceil \log_2(L+1) \rceil ) $ independent Gaussian noise terms, which itself will be Gaussian.  We then apply a union bound over all possible $\Delta_0$ items in $\omega^{(\ell)}$ and further a union bound over all possible rounds $L$ to get the following when $T = 1 + \sigma \cdot  \sqrt{ \lceil\log_2(L+1)\rceil} \cdot \Phi^{-1}(1-\tfrac{\delta}{\Delta_0\cdot L})$
    \[
    \Pr[\neg E] \leq \Delta_0 \cdot L \cdot \Pr[\mathrm{N}(0,  \lceil\log_2(L+1)\rceil \sigma^2) \geq T - 1] = \delta.
    \]

    \item[2.] 
    We will write the distribution of the Binary Mechanism evaluated on $\omega^{(1:L)}$ with only the common items with $\omega'^{(1:L)}$ as $P'$ and whose counts are positive.  We then have for all outcomes $y \in S$ that
    \[
    P(y) = \Pr[E] P'(y)
    \]
    Since $Q$ is the distribution for the neighboring input $\omega^{(1:L)}$ where $\omega^{(\ell)} = \emptyset$, we have $S$ is all outcomes so that $Q(y) = Q'(y)$ for all $y \in S$.
\item[3.]  Note that $P'$ and $Q'$ are simply a post-processing functions of the partial sums table with Gaussian noise added to each cell in the table. Hence, we need to consider the R\'enyi divergence for the partial sum tables on common items between $\omega^{(1:L)}$ and $\omega^{(1:L)}$.  We then have $D_\lambda(P'||Q) \leq \Delta_0 (\lceil \log_2(L+1) \rceil)/\sigma^2 \lambda$ and $D_\lambda(P'||Q) \leq \Delta_0(\lceil \log_2(L+1) \rceil)/\sigma^2 \lambda$ for all $\lambda>1$ because this is a randomized post processing (including additional noise terms) on the common items in each partial sum table in the same as the Binary Mechanism.

    \end{itemize}
Setting $\sigma = 1/\diffp$ completes the proof.

\end{proof}

\section{Conclusion}
We have presented a unified framework to prove that several different algorithms over unknown domain histograms are approximate CDP.  In many settings, practitioners want to have a way to incorporate private algorithms with minimal onboarding.  A major bottleneck for incorporating private algorithms into existing systems is requiring a fixed list of items that we want to release counts for.  Furthermore, products teams might be comfortable with noise added to counts, but not displaying counts for items that never appeared in the dataset.  We wanted to show how the privacy analyses of many existing DP algorithms can be unified by fixing neighboring datasets and considering not just outcomes that can occur in both neighboring inputs, but also the related distributions that can only generate these \emph{good} outcomes.  We think that approximate CDP provides the easiest way to combine these algorithms together to get tight privacy loss bounds, as the privacy analysis of many rely on improved pure CDP bounds rather than pure DP bounds.  For example, we showed how using a CDP analysis of the Laplace mechanism can improve the CDP privacy parameter by considering composition over differing counts, rather than relying on an $\ell_1$-sensitivity bound as would be the case for DP.  We can also use the tighter connection between Exponential Mechanisms and CDP, rather than using pure DP parameters of the Exponential Mechanism.  Lastly, the Gaussian mechanism does not satisfy a pure DP bound, so using CDP is a natural fit and converting to approximate DP would result in a lossy DP parameter.  We hope that this unified framework will help demystify some previous analyses and can be leveraged in designing future private algorithms.  

\section{Acknowledgements}
Special thanks for the helpful comments from David Durfee and Thomas Steinke that helped improve the quality of this survey.

\bibliography{bib2}
\bibliographystyle{abbrvnat}

\end{document}